\documentclass[12pt]{article}

\usepackage{amsmath,amssymb,amsthm}
\usepackage[authoryear,round]{natbib}
\usepackage{microtype}
\usepackage{setspace}
\usepackage[margin=1in]{geometry}
\usepackage[T1]{fontenc}
\usepackage{lmodern}
\usepackage{xcolor}
\usepackage[colorlinks=true,citecolor=blue,linkcolor=blue,urlcolor=blue]{hyperref}
\hypersetup{pdftitle={Recovering Weighted Shapley Values from Position Preferences},pdfauthor={Sinan Ertemel}}

\newtheorem{theorem}{Theorem}
\newtheorem{corollary}{Corollary}
\newtheorem{proposition}{Proposition}
\newtheorem{axiom}{Axiom}
\newtheorem{condition}{Condition}
\theoremstyle{definition}
\newtheorem{example}{Example}

\newcommand{\cG}{\mathcal{G}}
\newcommand{\bzero}{\mathbf{0}}
\newcommand{\supp}{\operatorname{supp}}

\begin{document}

\title{Recovering Weighted Shapley Values from Position Preferences\thanks{During the preparation of this manuscript, the author used generative AI tools to assist with drafting, language editing, and checking mathematical arguments. The author takes full responsibility for the content of the paper.}}
\author{Sinan Ertemel\thanks{Department of Economics, Istanbul Technical University, Ma\c{c}ka 34367 Istanbul, T\"urkiye. Correspondence: \href{mailto:ertemels@itu.edu.tr}{ertemels@itu.edu.tr}.}}
\date{}
\maketitle

\begin{abstract}\singlespacing
Weighted Shapley values represent asymmetry through positive weights and ordered priority classes. We give these parameters a preference interpretation in Roth's framework of lotteries over positions in transferable-utility games. Under common calibration, nullity, ordinary risk neutrality, unanimity efficiency, and unanimity positivity, position values in unanimity games are probability equivalents. Luce conditioning of these probabilities characterizes the Kalai--Samet family and is equivalent to partnership consistency. We separate conditioning into transitivity of revealed priority, pairwise determination of support, and invariance of positive odds. Two-player observations admit a weighted Shapley extension exactly when revealed priority is transitive and within-class odds satisfy a product rule. The extension is unique: pairwise comparisons identify the priority classes and weight ratios within each class. Larger unanimity games provide additional restrictions on supports and shares. We give constructive proofs, establish independence of the five axioms, and explain the scope and limits of the elicitation interpretation.
\end{abstract}



\section{Introduction}

Weighted Shapley values allow cooperative surplus to be divided asymmetrically. As axiomatized by \citet{KalaiSamet1987,KalaiSamet1988}, they split each coalition's Harsanyi dividend among its highest-priority members in proportion to positive weights. The weights and the priority ordering may represent bargaining strength, entitlement, or institutional standing. Consider an evaluator comparing positions in a joint project whose return requires every participant. The value assigned to a position can be elicited against a lottery that awards its holder the return alone. We ask which comparisons across such projects admit a weighted Shapley representation, and how the weights and priorities can be recovered from those comparisons.

We build on \citet{Roth1977SV}, who interprets the value of a game as the von Neumann--Morgenstern utility of holding a position in it. Roth's evaluator is neutral to ordinary risk, which makes positional utilities linear in the game, and neutral to strategic risk, which requires a position in an $r$-player unanimity game to be worth a $1/r$ chance of holding the dividend alone. Together with his remaining axioms, these assumptions yield the symmetric Shapley value \citep{Shapley1953}. We keep the structure of Roth's strategic-risk condition but leave the probabilities free. Under nullity, efficiency, and positivity on unanimity games, the evaluator's certainty equivalent for a position in a unanimity game becomes a probability equivalent: the evaluator behaves as if each coalition's dividend were awarded entirely to one member, with member $i$ of coalition $S$ winning with probability $q_i(S)$. The restrictions on these probabilities across coalitions will determine which forms of asymmetry are possible.

Theorem~\ref{thm:main} shows that, under the remaining axioms, the evaluator's valuation is a weighted Shapley value if and only if the winning probabilities obey Luce's choice axiom in its zero-support form: whenever a subcoalition receives positive probability within a larger coalition, winning probabilities in the subcoalition are the conditional probabilities from the larger one. A player has lower priority than another exactly when the evaluator regards that player's position in their two-player unanimity game as worthless, and within a priority block the pairwise odds reveal the corresponding weight ratios. The same condition is equivalent to the partnership consistency of \citet{KalaiSamet1987}, both on unanimity games and on all games, which ties the preference axiom to the value-level literature.

Proposition~\ref{prop:decomp} shows that Luce conditioning decomposes into three conditions: revealed priority must be transitive, the positions valued positively in any coalition must be exactly those not revealed to rank below another member, and the odds between two such positions must not depend on who else is present. Two-player data reveal candidate priorities and odds, which must themselves satisfy transitivity and a product rule. Proposition~\ref{prop:pairwise} proves that these restrictions are necessary and sufficient for pairwise observations to admit a unique weighted Shapley extension. Under maintained unanimity efficiency, one probability equivalent per unordered pair suffices to recover that extension. Every larger unanimity game then supplies an overidentifying test of its predicted support and shares (Corollary~\ref{cor:tests}). Further implications include regularity, in the sense that a position never gains when the coalition grows, and the irrelevance of lower-priority partners. These are restrictions on an evaluator's exact indifferences. Their use in an experiment would additionally require a protocol for valuing positions and a treatment of elicitation error. Finally, the five axioms are independent, and three players is the smallest setting in which conditioning has content (Proposition~\ref{prop:independence}).

Earlier preference foundations provide the starting point. \citet{Roth1977SV} obtains the symmetric Shapley value, and \citet{Roth1977BA} shows that alternative assumptions about bargaining ability lead instead to the Banzhaf--Coleman index; in both, all positions are treated alike. \citet{BlairMcLean1990} derive probabilistic values \citep{Weber1988} from subjective expected utility and characterize semivalues and the Shapley and Banzhaf values within that class. Here we retain efficiency and allow positions to be valued differently. Conditioning then yields the Kalai--Samet family: the ordered priority partition is uniquely identified, and the weights are identified up to a separate positive normalization within each block.

Weighted Shapley values have also been characterized directly as allocation rules. \citet{KalaiSamet1987} use efficiency, additivity, positivity on monotone games, the dummy axiom, and partnership consistency, and subsequent work modifies symmetry, marginality, or priority requirements \citep{Chun1991,NowakRadzik1995,Casajus2019,Casajus2026,BealEtAl2023}. At the level of the induced value, our axioms parallel theirs: ordinary risk neutrality corresponds to additivity, unanimity efficiency to efficiency, nullity with calibration to the dummy axiom, and conditioning to partnership consistency, while positivity is needed only on unanimity games. Our axioms are stated in terms of an evaluator's indifferences. The decomposition of conditioning and the pairwise extension result express the representation in terms of comparisons that an evaluator can make between positions and calibrated lotteries.

The connection between consistent dividend sharing and weighted Shapley values is already present in the literature. \citet{DerksHallerPeters2000} study sharing systems and Bayesian consistency within the selectope; \citet{BillotThisse2005} formulate related restrictions through Luce choice. \citet{Dehez2017} develops the dividend-sharing approach, and \citet{Dehez2026} discusses the earlier attribution of the consistency result. In random choice theory, \citet{CerreiaVioglioEtAl2021} characterize Luce's choice axiom without strict positivity by a rational support correspondence and proportional choice within that support. Their finite-domain representation provides the ordered classes and positive weights used here. We do not claim a new representation of Luce choice or a new connection between conditioning and weighted Shapley values. Our contribution is to formulate these restrictions on preferences over game positions, relate them to partnership consistency under the same assumptions, and give an explicit account of what pairwise indifferences identify and what larger-coalition observations can reject.

The paper proceeds as follows. Section~\ref{sec:model} introduces the preference framework and the unanimity axioms. Section~\ref{sec:char} states and proves the characterization. Section~\ref{sec:behavior} develops the pairwise decomposition, identification, testable restrictions, and elicitation. Section~\ref{sec:concl} concludes. Appendix~\ref{sec:indep} establishes independence of the axioms.

\section{Preferences and unanimity shares}\label{sec:model}

Let $N$ be finite and nonempty, and let $\cG$ be the vector space of transferable-utility games $v:2^N\to\mathbb{R}$ with $v(\varnothing)=0$. For nonempty $S\subseteq N$, let $u_S(T)=1$ if $S\subseteq T$ and zero otherwise. The unanimity games form a basis:
\begin{equation}\label{eq:mobius}
v=\sum_{\varnothing\ne S\subseteq N}\Delta_v(S)u_S,
\qquad
\Delta_v(S)=\sum_{T\subseteq S}(-1)^{|S|-|T|}v(T).
\end{equation}
The full vector-space domain permits scaling and sign reversal, both of which are used below.

An evaluator ranks finite lotteries over $\cG\times N$, where $(v,i)$ means position $i$ in $v$. Assume the von Neumann--Morgenstern axioms, so a pure-position index $U:\cG\times N\to\mathbb{R}$ evaluates lotteries by expectation \citep{HersteinMilnor1953}. For all $i,j,k\in N$, impose the common payoff calibration
\[
(\bzero,i)\sim(\bzero,j),
\qquad
(u_{\{i\}},i)\sim(u_{\{j\}},j)\succ(\bzero,k).
\]
All comparisons are made by the same evaluator. The calibration is therefore a restriction on one preference relation, rather than an interpersonal comparison of different players' utilities. It permits the global normalization $U(\bzero,i)=0$ and $U(u_{\{i\}},i)=1$ for every $i$. Write $\theta_i(v)=U(v,i)$ and $\theta(v)=(\theta_i(v))_{i\in N}$.

Calibration assigns the same value to controlling one unit alone, regardless of the player's identity. Positions may still be valued differently when obtaining the unit requires cooperation.

\begin{axiom}[Nullity, N]\label{ax:nullity}
For every $v\in\cG$ and $i\in N$, if $v(T\cup\{i\})=v(T)$ for all $T\subseteq N\setminus\{i\}$, then $(v,i)\sim(\bzero,i)$.
\end{axiom}

A player who changes no coalition's worth has a position worth the same as a position in the null game. In a unanimity game, this excludes outsiders from the dividend while allowing its indispensable members to be valued differently.

\begin{axiom}[Ordinary risk neutrality, OR]\label{ax:risk}
For every $i\in N$, $v,w\in\cG$, and $p\in[0,1]$,
\[
(pv+(1-p)w,i)\sim[p(v,i);(1-p)(w,i)].
\]
\end{axiom}
The expression in brackets denotes a lottery. The evaluator is indifferent between averaging coalition worths before evaluating position $i$ and facing a lottery over that position in the original games. Expected utility alone does not imply this indifference, since it does not specify the value of the averaged game. The axiom concerns uncertainty over games and leaves open how dependence on partners affects the value of a position.

The axiom makes $\theta_i$ affine. Since $\theta_i(\bzero)=0$, affinity gives $\theta_i(tv)=t\theta_i(v)$ for $t\in[0,1]$; for $t>1$, apply this identity to $tv$ with coefficient $1/t$. Midpoint affinity of $v$ and $-v$ gives $\theta_i(-v)=-\theta_i(v)$, and hence
$\theta_i(v+w)=2\theta_i((v+w)/2)=\theta_i(v)+\theta_i(w)$. Thus $\theta_i$ is linear; the full vector-space domain is what allows homogeneity for $t>1$ and $t<0$ to be derived.

For nonempty $S$, write $q_i(S)=\theta_i(u_S)$, $q(S)=(q_i(S))_{i\in N}$, and $q_S(T)=\sum_{j\in T}q_j(S)$ for $T\subseteq N$. We use $\supp q(S)=\{i\in N:q_i(S)>0\}$ when the shares are nonnegative. By (N), $q_i(S)=0$ outside $S$. The \emph{normalized certainty equivalent} of position $i$ in $u_S$ is the number $c$ satisfying $(u_S,i)\sim(c\,u_{\{i\}},i)$. By (OR) and calibration, $\theta_i(c\,u_{\{i\}})=c$, so this certainty equivalent exists, is unique, and equals $q_i(S)$. The three conditions below are therefore restrictions on observable indifferences.

\begin{axiom}[Unanimity efficiency, UE]\label{ax:efficiency}
For every nonempty $S\subseteq N$, $\sum_{i\in S}q_i(S)=1$.
\end{axiom}

The members of $S$ jointly control one unit, so their position values must sum to one. The entire value may nevertheless be assigned to a single member. Because it adds certainty equivalents of different positions, (UE) relies on the cross-position calibration above. Together with (N) and (OR), it is equivalent to full efficiency, $\sum_{i\in N}\theta_i(v)=v(N)$ for every $v\in\cG$.

\begin{axiom}[Unanimity positivity, UP]\label{ax:positivity}
For every nonempty $S\subseteq N$ and every $i\in S$, $q_i(S)\geq 0$.
\end{axiom}

Equivalently, $(u_S,i)\succsim(\bzero,i)$ whenever $i\in S$. A position needed to produce a positive unit is thus at least as valuable as a null position. Equality is allowed: a member may be indispensable yet receive zero value in the presence of higher-priority partners. Under the characterization below, strict positivity for every member would restrict the representation to a single priority block.

Under (N), (UE), and (UP), $q(S)$ is a probability distribution on the ``menu'' $S$, and
\[
(u_S,i)\sim[q_i(S)(u_{\{i\}},i);(1-q_i(S))(\bzero,i)].
\]
Thus, for $i\in S$, $q_i(S)$ is both the normalized certainty equivalent and the probability equivalent of position $i$ in $u_S$. Roth's strategic-risk neutrality is the special case $q_i(S)=1/|S|$.

Holding position $i$ in $u_S$ is worth the same as a $q_i(S)$ chance of holding the entire unit dividend alone. Linearity extends this interpretation to signed dividends. For each $S$, let $W_S\in S$ have distribution $q(S)$, and set $X_i(v)=\sum_{S\ni i}\Delta_v(S)\mathbf{1}_{\{W_S=i\}}$. Then $\mathbb{E}X_i(v)=\theta_i(v)$, and $(v,i)$ is indifferent to the finite lottery over $(X_i(v)u_{\{i\}},i)$. Negative dividends are assigned as liabilities. Only the marginal distributions of the $W_S$ matter for this equality; independence across coalitions is unnecessary. These draws represent the evaluator's indifferences and are not a proposed bargaining procedure.

For a subgroup $T$ of $S$, we can elicit position values directly in $u_T$ or condition the draw in $u_S$ on its winner belonging to $T$. We require the two evaluations to agree whenever the conditioning event has positive probability. Removing partners may then change a player's winning probability, but the odds between members with positive shares remain the same.

\begin{axiom}[Luce conditioning, LC]\label{ax:conditioning}
For every nonempty $S\subseteq N$ and every $\varnothing\ne T\subsetneq S$ with $q_S(T)>0$,
\begin{equation}
q_i(T)=\frac{q_i(S)}{q_S(T)}\qquad(i\in T).
\tag{LC}
\end{equation}
\end{axiom}
Under (UP), this is equivalent to $q_i(S)=q_S(T)q_i(T)$ for every $i\in T$. Indeed, if $q_S(T)=0$, nonnegativity forces $q_i(S)=0$ for every $i\in T$. With full support, this is Luce's familiar choice axiom; the stated zero-support form allows players receiving zero mass in $S$ to receive positive shares once a higher-priority tier is removed \citep{Luce1959,CerreiaVioglioEtAl2021}.

For example, if three members have shares $(1/2,1/3,1/6)$, conditioning on the first two requires shares $(3/5,2/5)$ in their two-player game: their odds remain $3:2$. By contrast, if their shares in the three-player game are $(0,0,1)$, the axiom places no restriction on how the first two divide the unit when considered alone. The event being conditioned on has probability zero. The first two players can therefore have their own priority ranking once the third is absent.

\section{Characterization}\label{sec:char}

A weight system is $\omega=(\lambda,\Sigma)$, where $\lambda\in\mathbb{R}_{++}^N$ and $\Sigma=(B_1,\ldots,B_m)$ is an ordered partition from lowest to highest priority. Put
\[
A_\omega(S)=S\cap B_{\max\{r:S\cap B_r\ne\varnothing\}}.
\]
The weighted Shapley value $\phi^\omega$ is the linear map defined on the unanimity basis by
\begin{equation}\label{eq:wsv}
\phi_i^\omega(u_S)=
\begin{cases}
\displaystyle\frac{\lambda_i}{\sum_{j\in A_\omega(S)}\lambda_j},&i\in A_\omega(S),\\[1.2ex]
0,&i\notin A_\omega(S).
\end{cases}
\end{equation}
Thus only members of the highest-priority block represented in $S$ share its unanimity dividend.

A nonempty coalition $P$ is a \emph{partnership} in $v$ when $v(R\cup L)=v(R)$ for every $L\subsetneq P$ and $R\subseteq N\setminus P$. Write $\theta_P(v)=\sum_{j\in P}\theta_j(v)$. Partnership consistency requires
\[
\theta_i(v)=\theta_i\!\left(\theta_P(v)u_P\right)\qquad(i\in P)
\]
whenever $P$ is a partnership in $v$.

The members of a partnership contribute only as a group; a proper subset adds nothing to a coalition outside $P$. Consistency requires that each member's value be unchanged when the original game is replaced by a unanimity game among the partners with worth equal to their total value. The theorem relates this way of dividing the partnership's value to conditioning.

\begin{theorem}[Preference characterization]\label{thm:main}
Under the von Neumann--Morgenstern and common-calibration assumptions, $\theta$ is a Kalai--Samet weighted Shapley value if and only if (N), (OR), (UE), (UP), and (LC) hold. Given the first four axioms, the following are equivalent: (LC), partnership consistency on unanimity games, and partnership consistency on all games.
\end{theorem}

\begin{proof}
Suppose first that the five axioms hold. By (OR), every $\theta_i$ is linear, so \eqref{eq:mobius} gives
\[
\theta_i(v)=\sum_{\varnothing\ne S\subseteq N}\Delta_v(S)q_i(S).
\]
Set $R_0=N$. Whenever $R_{r-1}\ne\varnothing$, define
\[
H_r=\supp q(R_{r-1}),\qquad
R_r=R_{r-1}\setminus H_r,
\qquad
\lambda_i=q_i(R_{r-1})\quad(i\in H_r).
\]
By (N), $H_r\subseteq R_{r-1}$; by (UE) and (UP), it is nonempty with positive weights. Thus the recursion terminates; let $m$ be the first index with $R_m=\varnothing$, so that $H_1,\ldots,H_m$ partition $N$ and $\lambda\in\mathbb{R}_{++}^N$. For nonempty $S$, let $r=\min\{k:S\cap H_k\ne\varnothing\}$. Then $S\subseteq R_{r-1}$ and $q_{R_{r-1}}(S)=\sum_{j\in S\cap H_r}\lambda_j>0$. If $S\subsetneq R_{r-1}$, (LC) gives
\[
q_i(S)=
\begin{cases}
\displaystyle\frac{\lambda_i}{\sum_{j\in S\cap H_r}\lambda_j},&i\in S\cap H_r,\\[1.2ex]
0,&i\notin S\cap H_r.
\end{cases}
\]
For $S=R_{r-1}$, the formula follows from the definitions and (UE), since $\sum_{j\in H_r}\lambda_j=1$. Thus \eqref{eq:wsv} holds for $\Sigma=(H_m,\ldots,H_1)$, and linearity yields $\theta=\phi^\omega$.

Conversely, suppose $\theta=\phi^\omega$. Linearity gives (OR), and \eqref{eq:wsv} gives (UE) and (UP). For (LC), let $\varnothing\ne T\subsetneq S$ with $q_S(T)>0$. Then $T$ meets $A_\omega(S)$, so the highest-priority block represented in $T$ is the one represented in $S$, and $A_\omega(T)=T\cap A_\omega(S)$. Dividing the shares in \eqref{eq:wsv} gives $q_i(S)/q_S(T)=\lambda_i/\sum_{j\in A_\omega(T)}\lambda_j=q_i(T)$ for $i\in A_\omega(T)$, and both sides vanish for $i\in T\setminus A_\omega(T)$. If $i$ is null in $v$, then $\Delta_v(S)=0$ for $S\ni i$, while \eqref{eq:wsv} gives $\phi_i^\omega(u_S)=0$ for $S\not\ni i$; hence \eqref{eq:mobius} gives (N).

It remains to connect conditioning and partnership consistency. A nonempty $P$ is a partnership in $u_S$ exactly when $P\subseteq S$ or $P\cap S=\varnothing$. Sufficiency is immediate. Otherwise, $P$ meets $S$ without being contained in it, so $L=P\cap S$ is a nonempty proper subset of $P$, and $R=S\setminus P\subseteq N\setminus P$ satisfies $u_S(R\cup L)=1>u_S(R)=0$.

For $P\subseteq S$, linearity turns partnership consistency into
\begin{equation}\label{eq:mult}
q_i(S)=q_S(P)q_i(P),\qquad i\in P.
\end{equation}
For proper $P\subsetneq S$, \eqref{eq:mult} is the multiplicative form of (LC). The cases $P=S$ and $P\cap S=\varnothing$ follow from (UE) and (N), respectively. This proves the equivalence on unanimity games.

Finally, let $P$ be a partnership in an arbitrary game $v$. Then $\Delta_v(S)=0$ whenever $\varnothing\ne S\cap P\subsetneq P$. Indeed, write $S=R\cup L$, where $R\subseteq N\setminus P$ and $\varnothing\ne L\subsetneq P$. Since partnership gives $v(A\cup B)=v(A)$ for all $A\subseteq R$ and $B\subseteq L$,
\[
\begin{aligned}
\Delta_v(S)
&=\sum_{A\subseteq R}\sum_{B\subseteq L}
(-1)^{|R|-|A|+|L|-|B|}v(A\cup B)\\
&=\sum_{A\subseteq R}(-1)^{|R|-|A|}v(A)
  \sum_{B\subseteq L}(-1)^{|L|-|B|}=0.
\end{aligned}
\]
The second equality uses partnership, and the last follows from $L\ne\varnothing$. Now fix $i\in P$. In the expansion $\theta_i(v)=\sum_S\Delta_v(S)q_i(S)$, coalitions not containing $i$ contribute nothing by (N), and coalitions containing $i$ but not $P$ meet $P$ properly and so have zero dividend. Hence
\[
\theta_i(v)
=\sum_{S\supseteq P}\Delta_v(S)q_i(S)
=q_i(P)\sum_{S\supseteq P}\Delta_v(S)q_S(P)
=q_i(P)\theta_P(v)
=\theta_i\!\left(\theta_P(v)u_P\right).
\]
The second equality uses multiplicative conditioning for $S\supsetneq P$ and (UE) for $S=P$. For the third, expand $\theta_P(v)=\sum_S\Delta_v(S)q_S(P)$: coalitions disjoint from $P$ have $q_S(P)=0$ by (N), and coalitions meeting $P$ properly have zero dividend, so only $S\supseteq P$ remain. The last equality is linearity. This proves partnership consistency on all games. The implication from full to unanimity-game consistency is immediate.
\end{proof}

\section{Behavioral content}\label{sec:behavior}

To relate conditioning to two-player comparisons, define revealed priority as follows. For distinct $i,j\in N$, say that $i$ is revealed to have \emph{at least as high priority} as $j$, written $i\succeq j$, if $q_i(\{i,j\})>0$; that is, the evaluator strictly prefers position $i$ in $u_{\{i,j\}}$ to position $i$ in the null game. Under (N), (UE), and (UP), the two shares in $u_{\{i,j\}}$ are nonnegative and sum to one, so $\succeq$ is complete on distinct players; extend it reflexively by setting $i\succeq i$. When both shares are positive, define the revealed odds $\rho_{ij}=q_i(\{i,j\})/q_j(\{i,j\})$, and set $\rho_{ii}=1$. For nonempty $S$, let $M(S)=\{i\in S:i\succeq j\text{ for every }j\in S\setminus\{i\}\}$ be the members not revealed to have lower priority than another member.

Three conditions relate these pairwise comparisons to evaluations of larger coalitions. The first concerns the priority relation itself; the other two specify how pairwise priorities and odds carry over to larger coalitions.

\begin{condition}[Transitive priority, TP]\label{cond:priority}
For all $i,j,k\in N$, if $i\succeq j$ and $j\succeq k$, then $i\succeq k$.
\end{condition}

If $i$ ranks at least as high as $j$, and $j$ at least as high as $k$, then $i$ cannot rank below $k$. A tie in this ranking means that both positions have positive value in their pairwise game; it does not require equal shares.

\begin{condition}[Pairwise support, PS]\label{cond:support}
For every nonempty $S\subseteq N$, $\supp q(S)=M(S)$.
\end{condition}

A member receives a positive share exactly when no partner outranks that member in a pairwise comparison. Under transitivity, these are the members of the highest priority class represented in the coalition. The condition determines the support, but not the relative sizes of the positive shares.

\begin{condition}[Odds invariance, OI]\label{cond:odds}
For every nonempty $S\subseteq N$ and all distinct $i,j\in\supp q(S)$, both $q_i(\{i,j\})$ and $q_j(\{i,j\})$ are positive, and $q_i(S)/q_j(S)=\rho_{ij}$.
\end{condition}

Among members who receive positive shares, relative value is unaffected by the presence of other partners. If one position is worth twice another in their pairwise game, it must remain worth twice as much in any larger coalition where both receive positive shares. Absolute shares may fall as other members enter. The positivity clause makes (OI) well-defined independently of (PS). Under (PS), it is automatic, since any two members of $\supp q(S)$ are mutually related by $\succeq$. This is the independence-of-irrelevant-alternatives property of \citet{Luce1959}, applied to the evaluator's winning probabilities.

\begin{proposition}[Pairwise decomposition]\label{prop:decomp}
Under the von Neumann--Morgenstern and common-calibration assumptions, and given (N), (OR), (UE), and (UP), condition (LC) holds if and only if (TP), (PS), and (OI) hold.
\end{proposition}

\begin{proof}
Suppose (LC) holds. By Theorem~\ref{thm:main}, $\theta=\phi^\omega$ for some $\omega=(\lambda,\Sigma)$, and by \eqref{eq:wsv}, $i\succeq j$ if and only if the block of $i$ is at least as high as the block of $j$. This relation is transitive, so (TP) holds. Moreover, $M(S)$ is the set of members of $S$ in the highest block represented in $S$, which is $A_\omega(S)=\supp q(S)$, so (PS) holds. Finally, for $i,j\in A_\omega(S)$, both $q_i(S)/q_j(S)$ and $\rho_{ij}$ equal $\lambda_i/\lambda_j$, so (OI) holds.

Conversely, suppose (TP), (PS), and (OI) hold. Since $\succeq$ is complete and transitive, $M(S)$ is the top indifference class of $\succeq$ within $S$. Let $\varnothing\ne T\subsetneq S$ with $q_S(T)>0$. By (PS), $T$ meets $M(S)$. We claim $M(T)=T\cap M(S)$. Members of $T\cap M(S)$ are ranked at least as high as all of $S\supseteq T$, so they lie in $M(T)$. Conversely, take $i\in M(T)$ and $a\in T\cap M(S)$. Then $i\succeq a$ and $a\succeq j$ for every $j\in S$, so $i\succeq j$ for every $j\in S$ by (TP), and $i\in M(S)$.

Now fix $a\in M(T)$. By (OI) applied to $S$, $q_k(S)=\rho_{ka}\,q_a(S)$ for every $k\in M(S)$, and hence, for $i\in M(T)$,
\[
\frac{q_i(S)}{q_S(T)}=\frac{q_i(S)}{\sum_{k\in M(T)}q_k(S)}=\frac{\rho_{ia}}{\sum_{k\in M(T)}\rho_{ka}},
\]
where the first equality uses (PS) and $M(T)=T\cap M(S)$. By (PS), (N), and (UE), $q(T)$ is a probability distribution supported on $M(T)$, and by (OI) applied to $T$, $q_k(T)=\rho_{ka}\,q_a(T)$ for $k\in M(T)$. Hence $q_i(T)=\rho_{ia}/\sum_{k\in M(T)}\rho_{ka}$, which agrees with the display. For $i\in T\setminus M(T)$, $i\notin M(S)$ by the claim, so $q_i(S)=0=q_i(T)$. Thus (LC) holds.
\end{proof}

The three conditions are independent given the other axioms. Take $N=\{1,2,3\}$ throughout, set each singleton share to one and all outsider shares to zero, and extend linearly from the unanimity basis. With all two-player shares equal to $1/2$ and $q(N)=(1/2,1/2,0)$, (TP) and (OI) hold but (PS) fails. With all two-player shares equal to $1/2$ and $q(N)=(1/2,1/3,1/6)$, (TP) and (PS) hold but (OI) fails. For (TP), let players $1$ and $2$ tie, let $2$ and $3$ tie, and let $3$ have priority over $1$, so that $q(\{1,2\})=(1/2,1/2,0)$, $q(\{2,3\})=(0,1/2,1/2)$, $q(\{1,3\})=(0,0,1)$, and $q(N)=(0,1/2,1/2)$. Then $M(N)=\{2,3\}$, and (PS) and (OI) hold, but $1\succeq2\succeq3$ while $3$ is revealed strictly above $1$.

\begin{corollary}[Identification]\label{cor:identification}
Under the conditions of Theorem~\ref{thm:main}, certainty equivalents in two-player unanimity games identify the ordered priority blocks, all within-block weight ratios, and hence the value on every game. For distinct $i,j$, player $i$ has higher priority than $j$ if and only if $q_i(\{i,j\})=1$; they share a block if and only if both pairwise shares are positive, in which case $\lambda_i/\lambda_j=\rho_{ij}$. The recovered system has one block if and only if $q_i(N)>0$ for every $i$. With the normalization $\sum_{i\in N}\lambda_i=1$, one then has $\lambda_i=q_i(N)$ and $q_i(S)=\lambda_i/\sum_{j\in S}\lambda_j$. Equal grand-coalition certainty equivalents therefore give the symmetric Shapley value.
\end{corollary}

\begin{proof}
By \eqref{eq:wsv}, a pair from different blocks assigns the entire unit to its higher-priority member. A pair from the same block assigns strictly positive shares in the ratio of the two weights. Thus the pairwise observations determine the blocks, their ordering, and every within-block weight ratio. Multiplying all weights in a block by the same positive constant changes no unanimity share. Conversely, identical pairwise observations leave no other freedom in the weights. Linearity then determines the value on every game. The assertions about a single block and equal shares follow by applying \eqref{eq:wsv} to $N$. For $|N|=1$, calibration already determines the unique value.
\end{proof}

Identification under the model differs from deciding whether a collection of pairwise observations is compatible with it. The following result gives the exact compatibility test. It concerns a complete set of two-player observations; no larger-coalition evaluations are assumed.

\begin{proposition}[Pairwise rationalizability]\label{prop:pairwise}
For every unordered pair $\{i,j\}$, suppose observed shares $b_i^{ij},b_j^{ij}\geq0$ sum to one. Define $i\succeq_b j$ by $b_i^{ij}>0$ for distinct players and extend the relation reflexively. Whenever both shares are positive, let $r_{ij}=b_i^{ij}/b_j^{ij}$, and put $r_{ii}=1$. There exists a weighted Shapley value matching all these observations if and only if:
\begin{enumerate}
\item[(i)] $\succeq_b$ is transitive;
\item[(ii)] $r_{ij}r_{jk}=r_{ik}$ for every three distinct players with mutually positive pairwise shares.
\end{enumerate}
When it exists, the value is unique. Its ordered partition is the partition into indifference classes of $\succeq_b$. In each class $B$, choosing any $a\in B$ gives the uniquely normalized weights
\begin{equation}\label{eq:recovery}
\lambda_i=\frac{r_{ia}}{\sum_{j\in B}r_{ja}}\qquad(i\in B),
\end{equation}
with $\sum_{i\in B}\lambda_i=1$.
\end{proposition}

\begin{proof}
Necessity follows from \eqref{eq:wsv}: the block order is transitive and, within a block, $r_{ij}=\lambda_i/\lambda_j$.

For sufficiency, completeness follows from the pairwise shares summing to one. Together with (i), it makes $\succeq_b$ a weak order. List its indifference classes from lowest to highest and choose $a$ in each class $B$. Define the positive weights by \eqref{eq:recovery}. For $i,j\in B$, condition (ii) gives $r_{ia}=r_{ij}r_{ja}$ when $i,j,a$ are distinct. The same identity follows from $r_{aa}=1$ and $r_{ij}r_{ji}=1$ when an index repeats. Hence $\lambda_i/\lambda_j=r_{ij}$. The weighted Shapley value with these weights and classes matches a within-class pair because its shares sum to one and have the observed ratio. It matches a between-class pair because the higher member's observed share is one and the lower member's is zero. Every matching value must have these same classes and ratios, so \eqref{eq:wsv} and linearity imply uniqueness. Formula~\eqref{eq:recovery} is consequently independent of the choice of $a$. A singleton class has weight one, covering also $|N|=1$.
\end{proof}

\begin{example}\label{ex:recovery}
Let $N=\{1,2,3,4\}$. Suppose the observed shares are $(2/3,1/3)$ for the pair $\{1,2\}$ and $(3/4,1/4)$ for $\{3,4\}$, while in every pair with one member from each group the member of $\{1,2\}$ receives one. The recovered partition is $\Sigma=(\{3,4\},\{1,2\})$, with normalized weights $(\lambda_1,\lambda_2)=(2/3,1/3)$ and $(\lambda_3,\lambda_4)=(3/4,1/4)$. Thus $q(\{1,2,3\})=(2/3,1/3,0,0)$, whereas removing both higher-priority players gives $q(\{3,4\})=(0,0,3/4,1/4)$. For the game $v=4u_{\{1,2\}}+8u_{\{3,4\}}-2u_{\{1,2,3\}}+3u_{\{1,4\}}$, linearity gives $\theta(v)=(13/3,2/3,6,2)$, whose coordinates sum to $v(N)=13$. The pairwise observations therefore determine both the treatment of priority and the allocation of signed dividends.
\end{example}

\begin{corollary}[Testable restrictions]\label{cor:tests}
Under the conditions of Theorem~\ref{thm:main}, the evaluator's certainty equivalents satisfy the following.
\begin{enumerate}
\item[(a)] \emph{Product rule.} If $i,j,k$ are pairwise revealed to share a block, then $\rho_{ij}\rho_{jk}=\rho_{ik}$.
\item[(b)] \emph{Regularity.} For $i\in S$ and $k\notin S$, $q_i(S\cup\{k\})\le q_i(S)$.
\item[(c)] \emph{Irrelevance of lower-priority partners.} For nonempty $S$ and $k\notin S$, if $q_k(S\cup\{k\})=0$, then $q_i(S\cup\{k\})=q_i(S)$ for every $i\in S$.
\item[(d)] \emph{Overidentification.} Every certainty equivalent in a unanimity game with three or more members is determined by the two-player certainty equivalents.
\end{enumerate}
\end{corollary}

\begin{proof}
(a) The representation gives $\rho_{ij}\rho_{jk}=(\lambda_i/\lambda_j)(\lambda_j/\lambda_k)=\lambda_i/\lambda_k=\rho_{ik}$. (b) Put $S'=S\cup\{k\}$. If $q_{S'}(S)=0$, (UP) gives $q_i(S')=0\le q_i(S)$. Otherwise, (LC) gives $q_i(S')=q_{S'}(S)q_i(S)$, and (UE) and (UP) imply $0<q_{S'}(S)\le1$, proving the inequality. (c) By (UE), $q_k(S\cup\{k\})=0$ implies $q_{S\cup\{k\}}(S)=1$; applying (LC) gives the claim. (d) This is Corollary~\ref{cor:identification}.
\end{proof}

To elicit a probability equivalent for $(u_S,i)$, find the probability $p$ at which the evaluator is indifferent between that position and the lottery $[p(u_{\{i\}},i);(1-p)(\bzero,i)]$. These are comparisons by a single evaluator across positions. They need not describe the beliefs or bargaining behavior of the players themselves.

The two-player observations test transitivity and the product rule on triples of players. Proposition~\ref{prop:pairwise} shows that passing these tests determines a candidate weighted Shapley value. With $n$ players and maintained (UE), there is one independent probability equivalent for each of the $\binom n2$ pairs. Eliciting both positions gives $n(n-1)$ observations and also checks (UE) on pairs. Observations from any of the $2^n-1-n-\binom n2$ unanimity games with at least three members can then be compared with the candidate's predicted support and shares. A failure of (TP) rejects a priority ranking; a failure of (PS) rejects its prediction of which positions receive positive value; and a failure of (OI) rejects the constancy of relative positive shares.

Passing the tests on three-player games does not validate evaluations of all larger games. For example, with four players, assign uniform shares to every proper coalition but set $q(N)=(1/2,1/6,1/6,1/6)$. All observations on coalitions of size at most three agree with the ordinary Shapley value, yet (LC) fails on $N$. Moreover, extrapolation from unanimity games to arbitrary games relies on (OR) and (N); these assumptions cannot be tested by unanimity observations alone.

The statements above concern exact indifferences. An experiment would need to specify what holding a position entails and how its eventual payoff is determined, as well as account for elicitation error. In particular, distinguishing a zero share from a small positive share is necessary to identify a strict priority boundary. The paper provides population restrictions and a recovery formula, rather than an incentive-compatible elicitation mechanism or a statistical test.

\section{Conclusion}\label{sec:concl}

Preferences over positions provide a way to identify the asymmetries represented by weighted Shapley values. Under the calibration, nullity, risk-neutrality, efficiency, and positivity assumptions, Luce conditioning characterizes the Kalai--Samet family. Two-player unanimity games identify the ordered priority blocks and the weight ratios within each block. Evaluations of larger unanimity games must then agree with the supports and shares predicted by those pairwise comparisons. The equivalence with partnership consistency connects this preference characterization to the existing allocation-rule approach, while the decomposition of conditioning distinguishes restrictions on priority from restrictions on relative shares.

\appendix
\section*{Appendix}
\section{Independence of the axioms}\label{sec:indep}

\begin{proposition}[Independence]\label{prop:independence}
The five axioms in Theorem~\ref{thm:main} are logically independent for $N=\{1,2,3\}$. Three players is the minimal size for joint independence: when $|N|\le 2$, (LC) is implied by calibration.
\end{proposition}

\begin{proof}
For minimality, (LC) is vacuous when $|N|=1$. When $|N|=2$, its only instances have the form $T=\{i\}\subsetneq S=\{i,j\}$ with $q_i(S)>0$, where it requires $q_i(\{i\})=1$, which is calibration.

Each construction below preserves common calibration and induces vNM preferences by evaluating lotteries at the expectation of $\theta$. When $q$ is specified, set outsider shares to zero and extend linearly from the unanimity basis; this extension satisfies (N) and (OR). Let $\phi$ be the ordinary Shapley value.

To establish independence of (N), let $\theta_i(v)=\phi_i(v)+\varepsilon v(N\setminus\{i\})$, $\varepsilon>0$. Calibration is preserved because $u_{\{i\}}(N\setminus\{i\})=0$, and member unanimity shares remain uniform, so (OR), (UE), (UP), and (LC) hold. But $i$ is null in $u_{\{j\}}$ and receives $\varepsilon$ when $i\ne j$. This $\theta$ also violates full efficiency, consistent with the equivalence between (UE) and efficiency holding only in the presence of (N).

To establish independence of (OR), for each $i$, fix $j(i)\in N\setminus\{i\}$ and set
\[
\theta_i(v)=\phi_i(v)+h\!\left(\Delta_v(\{i\}),\Delta_v(\{i,j(i)\})\right),
\qquad
h(a,b)=\begin{cases}a|b|/(|a|+|b|),&(a,b)\ne(0,0),\\0,&(a,b)=(0,0).
\end{cases}
\]
The perturbation is positively homogeneous but nonadditive and vanishes on every unanimity game and whenever $i$ is null. The other axioms hold. Yet $\theta_i(u_{\{i\}}+u_{\{i,j(i)\}})=2$, whereas $u_{\{i\}}+u_{\{i,j(i)\}}$ is the equal mixture of $2u_{\{i\}}$ and $2u_{\{i,j(i)\}}$, and the equal lottery over $(2u_{\{i\}},i)$ and $(2u_{\{i,j(i)\}},i)$ has utility $3/2$.

To establish independence of (UE), for every nonempty proper $S\subsetneq N$, set $q_i(S)=1/|S|$ for $i\in S$, and set $q_i(N)=2/3$ for every $i\in N$. Uniform shares are coherent within proper carriers; for $\varnothing\ne T\subsetneq N$, $q_N(T)=2|T|/3$ and conditioning again gives $q_i(T)=1/|T|$. Yet the grand shares sum to two.

To establish independence of (UP), for $a=(4,-1,-2)$, let $q_i(S)=a_i/\sum_{j\in S}a_j$ for $i\in S$. All nonempty coalition sums are nonzero, and whenever $q_S(T)>0$, (LC) follows by cancellation. But $q_2(N)=-1$.

To establish independence of (LC), use uniform proper-coalition shares and $q(N)=(1/2,1/3,1/6)$. Conditioning on $\{1,2\}$ gives $(3/5,2/5)$ instead of $(1/2,1/2)$.
\end{proof}

\section*{Declaration of generative AI and AI-assisted technologies in the manuscript preparation process}
Generative AI tools assisted with drafting and language editing, literature searches, and scrutiny of mathematical arguments. Responsibility for the arguments, citations, and final content rests with the author.

\begingroup\singlespacing
\bibliographystyle{plainnat}
\bibliography{references}

@article{BealEtAl2023,
  author  = {B{\'e}al, Sylvain and Ferri{\`e}res, Sylvain and Navarro-Ramos, Adriana and Solal, Philippe},
  title   = {Axiomatic Characterizations of the Family of Weighted Priority Values},
  journal = {International Journal of Economic Theory},
  year    = {2023},
  volume  = {19},
  number  = {4},
  pages   = {787--816},
  doi     = {10.1111/ijet.12375}
}

@article{BlairMcLean1990,
  author  = {Blair, Douglas H. and McLean, Richard P.},
  title   = {Subjective Evaluations of {$n$}-Person Games},
  journal = {Journal of Economic Theory},
  year    = {1990},
  volume  = {50},
  number  = {2},
  pages   = {346--361},
  doi     = {10.1016/0022-0531(90)90007-7}
}

@article{BillotThisse2005,
  author  = {Billot, Antoine and Thisse, Jacques-Fran{\c{c}}ois},
  title   = {How to Share When Context Matters: The {M{\"o}bius} Value as a Generalized Solution for Cooperative Games},
  journal = {Journal of Mathematical Economics},
  year    = {2005},
  volume  = {41},
  number  = {8},
  pages   = {1007--1029},
  doi     = {10.1016/j.jmateco.2004.12.008}
}

@article{CerreiaVioglioEtAl2021,
  author  = {Cerreia-Vioglio, Simone and Lindberg, Per Olov and Maccheroni, Fabio and Marinacci, Massimo and Rustichini, Aldo},
  title   = {A Canon of Probabilistic Rationality},
  journal = {Journal of Economic Theory},
  year    = {2021},
  volume  = {196},
  pages   = {105289},
  doi     = {10.1016/j.jet.2021.105289}
}

@article{Dehez2017,
  author  = {Dehez, Pierre},
  title   = {On {Harsanyi} Dividends and Asymmetric Values},
  journal = {International Game Theory Review},
  year    = {2017},
  volume  = {19},
  number  = {3},
  pages   = {1750012},
  doi     = {10.1142/S0219198917500128}
}

@article{DerksHallerPeters2000,
  author  = {Derks, Jean and Haller, Hans and Peters, Hans},
  title   = {The Selectope for Cooperative Games},
  journal = {International Journal of Game Theory},
  year    = {2000},
  volume  = {29},
  number  = {1},
  pages   = {23--38},
  doi     = {10.1007/s001820050003}
}

@article{Chun1991,
  author  = {Chun, Youngsub},
  title   = {On the Symmetric and Weighted {Shapley} Values},
  journal = {International Journal of Game Theory},
  year    = {1991},
  volume  = {20},
  number  = {2},
  pages   = {183--190},
  doi     = {10.1007/BF01240278}
}

@book{Luce1959,
  author    = {Luce, R. Duncan},
  title     = {Individual Choice Behavior: A Theoretical Analysis},
  publisher = {John Wiley \& Sons},
  address   = {New York},
  year      = {1959}
}

@article{NowakRadzik1995,
  author  = {Nowak, Andrzej S. and Radzik, Tadeusz},
  title   = {On Axiomatizations of the Weighted {Shapley} Values},
  journal = {Games and Economic Behavior},
  year    = {1995},
  volume  = {8},
  number  = {2},
  pages   = {389--405},
  doi     = {10.1016/S0899-8256(05)80008-9}
}

@article{HersteinMilnor1953,
  author  = {Herstein, I. N. and Milnor, John},
  title   = {An Axiomatic Approach to Measurable Utility},
  journal = {Econometrica},
  year    = {1953},
  volume  = {21},
  number  = {2},
  pages   = {291--297},
  doi     = {10.2307/1905540}
}

@article{KalaiSamet1987,
  author  = {Kalai, Ehud and Samet, Dov},
  title   = {On Weighted {Shapley} Values},
  journal = {International Journal of Game Theory},
  year    = {1987},
  volume  = {16},
  number  = {3},
  pages   = {205--222},
  doi     = {10.1007/BF01756292}
}

@incollection{KalaiSamet1988,
  author    = {Kalai, Ehud and Samet, Dov},
  title     = {Weighted {Shapley} Values},
  booktitle = {The Shapley Value: Essays in Honor of Lloyd S. Shapley},
  editor    = {Roth, Alvin E.},
  publisher = {Cambridge University Press},
  year      = {1988},
  pages     = {83--100},
  doi       = {10.1017/CBO9780511528446.007}
}

@incollection{Weber1988,
  author    = {Weber, Robert James},
  title     = {Probabilistic Values for Games},
  booktitle = {The Shapley Value: Essays in Honor of Lloyd S. Shapley},
  editor    = {Roth, Alvin E.},
  publisher = {Cambridge University Press},
  year      = {1988},
  pages     = {101--120},
  doi       = {10.1017/CBO9780511528446.008}
}

@article{Casajus2019,
  author  = {Casajus, Andr{\'e}},
  title   = {Relaxations of Symmetry and the Weighted {Shapley} Values},
  journal = {Economics Letters},
  year    = {2019},
  volume  = {176},
  pages   = {75--78},
  doi     = {10.1016/j.econlet.2018.12.031}
}

@article{Casajus2026,
  author  = {Casajus, Andr{\'e}},
  title   = {Related Characterizations of the {Shapley} Value and the Weighted {Shapley} Values via Relaxations of Differential Marginality},
  journal = {Discrete Applied Mathematics},
  year    = {2026},
  volume  = {378},
  pages   = {136--150},
  doi     = {10.1016/j.dam.2025.06.069}
}

@article{Roth1977BA,
  author  = {Roth, Alvin E.},
  title   = {Bargaining Ability, the Utility of Playing a Game, and Models of Coalition Formation},
  journal = {Journal of Mathematical Psychology},
  year    = {1977},
  volume  = {16},
  number  = {2},
  pages   = {153--160},
  doi     = {10.1016/0022-2496(77)90064-5}
}

@article{Roth1977SV,
  author  = {Roth, Alvin E.},
  title   = {The {Shapley} Value as a von {Neumann}--{Morgenstern} Utility},
  journal = {Econometrica},
  year    = {1977},
  volume  = {45},
  number  = {3},
  pages   = {657--664},
  doi     = {10.2307/1911680}
}

@incollection{Shapley1953,
  author    = {Shapley, Lloyd S.},
  title     = {A Value for $n$-Person Games},
  booktitle = {Contributions to the Theory of Games II},
  editor    = {Kuhn, Harold W. and Tucker, Albert W.},
  series    = {Annals of Mathematics Studies},
  volume    = {28},
  publisher = {Princeton University Press},
  address   = {Princeton, NJ},
  year      = {1953},
  pages     = {307--317},
  doi       = {10.1515/9781400881970-018}
}

@techreport{Dehez2026,
  author      = {Dehez, Pierre},
  title       = {Comment on ``How to Share When Context Matters: The {M{\"o}bius} Value as a Generalized Solution for Cooperative Games'' by {A. Billot} and {J. Thisse} (2005)},
  institution = {Universit{\'e} catholique de Louvain, CORE},
  type        = {LIDAM Discussion Paper CORE},
  number      = {2026/05},
  year        = {2026}
}
\endgroup

\end{document}